\documentclass[letterpaper, 10 pt, conference]{ieeeconf}  
\IEEEoverridecommandlockouts
\usepackage{balance}
\usepackage{color}
\usepackage{caption}
\usepackage{enumerate}
\usepackage{cite}
\usepackage{empheq}
\usepackage{mathrsfs}
\usepackage{flushend}

\usepackage{mathtools}

\usepackage{tikz}
\usepackage{circuitikz}

\usepackage[font=small,labelfont=bf]{caption}

\renewcommand{\natural}{{\mathbb{N}}}

\newcommand{\integernonnegative}{\ensuremath{\mathbb{Z}}_{\ge 0}}
\newcommand{\real}{\ensuremath{\mathbb{R}}}

\newcommand{\until}[1]{\{1,\dots, #1\}}

\newcommand{\supscr}[2]{#1^{\textup{#2}}}

\newcommand{\diag}[1]{\operatorname{diag}(#1)}

\newcommand{\rank}{\operatorname{rank}}

\newcommand{\graph}{G}

\newcommand{\edges}{E}

\newcommand{\neighbors}{\mathcal{N}}

\renewcommand{\epsilon}{\varepsilon}

\usepackage{psfrag}

\usepackage{amsmath}
\usepackage{cases}
\usepackage{subfig}


\makeatletter
\pgfcircdeclarebipole{}{\ctikzvalof{bipoles/interr/height 2}}{spst}{\ctikzvalof{bipoles/interr/height}}{\ctikzvalof{bipoles/interr/width}}{

    \pgfsetlinewidth{\pgfkeysvalueof{/tikz/circuitikz/bipoles/thickness}\pgfstartlinewidth}

    \pgfpathmoveto{\pgfpoint{\pgf@circ@res@left}{0pt}}
    \pgfpathlineto{\pgfpoint{.6\pgf@circ@res@right}{\pgf@circ@res@up}}
    \pgfusepath{draw}
}

\def\pgf@circ@spst@path#1{\pgf@circ@bipole@path{spst}{#1}}
\tikzset{switch/.style = {\circuitikzbasekey, /tikz/to path=\pgf@circ@spst@path, l=#1}}
\tikzset{spst/.style = {switch = #1}}
\makeatother

\makeatletter
\let\proof\@undefined                        
\let\endproof\@undefined                  
\makeatother
\usepackage{graphicx,amssymb,amstext,amsmath,amsthm}

\usepackage{float}
\usepackage[bookmarks=true]{hyperref}
\usepackage{algorithm,algpseudocode}

\usepackage{multirow}
\usepackage{stfloats}

\renewcommand{\graph}{\mathcal{G}}
\newcommand{\nodes}{\mathcal{V}}
\renewcommand{\edges}{\mathcal{E}}
\newcommand{\interval}{\mathcal{I}}

\newtheorem{prop}{Proposition} 

\newtheorem{thm}{Theorem}
	\newtheorem{assumption}{Assumption}
\newtheorem{lem}{Lemma}
\newtheorem{defn}{Definition}
\newtheorem{rem}{Remark}
\newtheorem{problem}{Problem}

\let\oldReturn\Return
\renewcommand{\Return}{\State\oldReturn}

\usepackage{setspace}
\let\oldbibliography\thebibliography
\renewcommand{\thebibliography}[1]{%
  \oldbibliography{#1}%
}

\newcommand{\marginresolved}[1]{}

\newcommand{\scott}[1]{{\color{black} #1}}
\newcommand{\mmo}[1]{{\color{black} #1}}
\newcommand{\kj}[1]{{\color{black} #1}}

\newcommand{\dsisoalgo}{\textsc{Distributed Simultaneous Input \&
    State Interval Observer}\;}
\newcommand{\dsiso}{\textsc{DSISO}}

\usepackage{xargs}
\newcommandx{\tightmath}[5][1=0,2=0,3=0,4=0]{{%
    \thickmuskip=#1mu
    \medmuskip=#2mu
    \thinmuskip=#3mu
    \arraycolsep=#4pt
    #5}}

\begin{document}

\title{\LARGE \bf Distributed Resilient Interval Observers for
  Bounded-Error LTI Systems Subject to False Data Injection
  Attacks} 

\author{%
  Mohammad Khajenejad$^*$, Scott Brown$^*$, and Sonia Mart{\'\i}nez \\
  \thanks{$^*$Equal contribution. M. Khajenejad, S. Brown and S. Mart{\'\i}nez are with the
    Mechanical and Aerospace Engineering Department of the Jacobs
    School of Engineering, University of California, San Diego, La
    Jolla, San Diego, CA, USA. (e-mail: \{mkhajenejad, sab007,
    soniamd\}@ucsd.edu.)}  \thanks{This work is partially funded by
    ARL and ONR grants W911NF-23-2-0009 and N00014-19-1-2471.}
}

\maketitle
\thispagestyle{empty}
\pagestyle{empty}
{
\begin{abstract}
  This paper proposes a novel distributed interval-valued simultaneous
  state and input \kj{observer} for linear time-invariant (LTI)
  systems \kj{that are} subject to attacks or unknown inputs, \kj{injected both on their
    sensors and actuators}. Each agent in the network leverages a
  singular value decomposition (SVD) based transformation \kj{to decompose its observations into two components, one of them unaffected by the attack signal, which helps} to obtain
  \emph{local} interval estimates of the state and unknown input \kj{and} then
  uses intersection to compute the best \kj{interval} estimate among neighboring
  nodes. We show that the computed intervals are guaranteed to contain
  the true state and input trajectories, and we provide conditions
  under which the observer is stable. Furthermore, we provide a method
  for designing stabilizing gains that minimize an upper bound on the
  worst-case steady-state observer error. We demonstrate our
  algorithm on an IEEE 14-bus power system.
\end{abstract}

\section{Introduction}
The control of Cyber-Physical Systems (CPS) relies on the tight
integration of various computational, communication, and sensor
components that interact with each other and \kj{with} the physical world in a
complex way. Applications of CPS are broad and include, to name a few,
industrial infrastructures~\cite{cheng2018industrial}, power grids
\cite{zhao2019dse}, and intelligent transportation systems
\cite{sun2017secure}. In such \emph{safety-critical} systems, serious
detriment can occur in case of malfunction or if jeopardized by
malicious attackers \cite{ukraine.2016}.
One of the most serious types of attacks consists of false-data
injection, by which counterfeit data signals are injected into the
actuator signals and sensor measurements by strategic and/or malicious
agents. Such attacks are not well-modeled by  zero-mean, Gaussian
white noises nor by signals with known bounds, given their strategic
nature. \kj{On the other hand, most of the centralized approaches to state estimation are computationally expensive, especially for realistic high-dimensional CPS.}
Consequently, reliable \kj{distributed state and unknown input estimation algorithms} 
are indispensable for the sake of resilient control, attack
identification, and mitigation.

Motivated by this, several estimation algorithms have
been proposed, in which a central entity seeks to estimate both the
system state and the unknown disturbance (input). In the context where the
noise signals are Gaussian and white, a large body of work proposed
different designs for joint input and state estimation via
e.g., minimum variance unbiased estimation
\cite{Yong.Zhu.ea.Automatica16}, modified double-model adaptive
estimation \cite{lu2016framework}, or robust regularized least square
approaches \cite{abolhasani2018robust}. However, these approaches are
not applicable in the context of attack-resilient \emph{bounded error}
worst-case estimation against false data injection attacks, where no
information about the distribution of uncertainties is available. In
such a setting, numerous approaches were proposed for deterministic
systems \cite{Pasqualetti.2013}, stochastic systems
\cite{kim2016attack}, and bounded-error systems 
\cite{nakahira2015dynamic,pajic2015attack,yong2016robust}.
These approaches typically yield point estimates, i.e, the most
likely or best single estimate, as opposed to set-valued estimates.

Set-valued estimates have the advantage of providing hard accuracy bounds, which
are important to guarantee safety \cite{yong2018simultaneous,blanchini2012convex,khajenejad2021guaranteed,khajenejad2022direct}.  In
addition, the use of \emph{fixed-order} set-valued methods can help
decrease the complexity of optimal observers
\cite{milanese1991optimal}, which grows with time.  Hence,
fixed-order centralized set-valued observers were presented for
different classes of systems
\cite{yong2018simultaneous,khajenejad2019simultaneous,khajenejadasimultaneous,khajenejad2021intervalACC,ellero2019unknown,khajenejad2022h,pati2022l},
that simultaneously find bounded sets of compatible states and unknown
inputs.  However, these algorithms do not scale well in a networked
setting as the size of the network increases.  This motivates the
design of \textit{distributed input and state filters}, which have
typically focused on systems with stochastic disturbances \cite{lu2013,ashari2012}.
While these methods are more scalable and robust to communication
failures than their centralized counterparts, they generally have
comparatively worse estimation error. \kj{Moreover, they are not applicable in bounded-error settings where no information about the stochastic characteristics of noise/disturbance is available. With that in mind, in our previous work \cite{khajenejad2022distributed} we provided a distributed algorithm to synthesize interval observers for bounded-error LTI systems, without considering unknown input signals (attacks). In this work we aim to extend our design in \cite{khajenejad2022distributed} to address \emph{resiliency} against false data injection attacks, i.e., to synthesize distributed interval observers in the bounded-error settings that are stable and correct in the presence of unknown input/attack signals.}

\textit{Contributions. } This work aims to bridge the gap between
distributed resilient estimation algorithms and interval observer
design approaches in bounded-error settings and subject to completely
unknown and/or distribution-free inputs (attacks). In other words,
leveraging the notion of ``collective positive detectability over
neighborhoods'' (CPDN), we provide a distributed algorithm that
simultaneously synthesizes tractable and computationally efficient
interval-valued estimates for states and unknown inputs of bounded
error LTI systems, whose sensors and actuators are subject to false
data injection attacks. We provide conditions for the stability of our
proposed observer, which is shown to minimize a computed upper bound
for the observer error interval widths. 

\section{Preliminaries}

This section introduces basic notation, preliminary concepts, and
graph-theoretic notions used throughout the paper.

\textit{Notation.}  Let $\natural$, $\integernonnegative$,
$\mathbb{R}_{\geq 0}$, denote the sets of natural, nonnegative integer
and real numbers, respectively.  Similarly, $\real^n$ and
$\real^{n \times p}$ denote the $n$-dimensional Euclidean space, and
the set of $n \times p$ matrices, respectively. Given
$A^1,\dots,A^N \in \real^{n \times n}$,
$\diag{A^1,\dots,A^N} \in \real^{nN \times nN}$ is the block-diagonal
matrix with block-diagonal elements $A^i,i\in \until{N}$. For
$M \in \real^{n \times p}$, $M_i$ and $M_{ij}$ denote the
$\supscr{i}{th}$ row and the $\supscr{(i,j)}{th}$ entry of $M$,
respectively. Furthermore, we define $M^+ \in \real^{n \times p}$,
such that $M^+_{ij} \triangleq \max \{M_{ij}, 0 \}$,
$M^-\triangleq M^+-M$, and $|M|\triangleq M^++M^-$. All inequalities
$\leq$, $\geq$, as well as $\max$ and $\min$, are considered
element-wise. Given $M \in \real^{n\times n}$, $\rho(M)$ is used to
denote the spectral radius of $M$.  A multi-dimensional interval is
{denoted as
  $\interval \triangleq [\underline{s},\overline{s}] \subset
  \real^n$}, and is the set of vectors $x \in \real^n$ such that
$\underline{s} \le x \le \overline{s}$.
\begin{prop}\cite[Lemma 1]{efimov2013interval}\label{prop:bounding}
  Let $A \in \real^{p \times n}$ and $\underline{x} \leq x \leq
  \overline{x} \in \real^n$. Then, $A^+\underline{x}-A^{-}\overline{x}
  \leq Ax \leq A^+\overline{x}-A^{-}\underline{x}$. As a corollary, if
  $A$ is non-negative, $A\underline{x} \leq Ax \leq A\overline{x}$.
\end{prop}
\textit{Graph-theoretic Notions.}  A directed graph
$\graph = (\nodes, \edges)$ is a set of nodes
$\nodes \triangleq \until{N}$ and edges
$\edges \subseteq \nodes \times \nodes$.  The neighbors of node $i$,
denoted $\neighbors_i$, is the set of all nodes $j$ for which
$(j,i)\in\edges$.  We will assume that $i \in \neighbors_i$.

\section{Problem Formulation} \label{sec:Problem}
\noindent\textbf{\textit{System Assumptions.}}
Consider a multi-agent system \mmo{(MAS)} consisting of $N$ agents,
which interact over a time-invariant communication graph $\graph$
\footnote{\scott{Later, the structure of the graph will play an important role
  in the satisfaction of Assumption~\ref{ass:col_pos_det_neigh}. We don't,
  however, require any specific assumptions about connectivity of the graph.
  Loosely speaking, if individual nodes
  have access to many complementary measurements, the graph need not
  be ``well connected''. The converse is also true, that a ``well
  connected'' graph can overcome a lack of individual measurements.}}. The
agents are able to obtain individual measurements of a target
system as described by the following LTI dynamics:
\begin{align}\label{eq:system}
\begin{array}{ll}
  \mathcal{P}:
  \begin{cases}
    x_{k+1} = Ax_k+Bw_k+Gd_k, &  \\
    y^i_k = C^ix_k+D^iv^i_k+H^id_k, & \hspace{-.1cm}i \in \nodes, k \in
    \integernonnegative,
  \end{cases}
\end{array}
\end{align}
where $x_k \in \real^{n}$ is the continuous state of the target
system, $d_k \in \mathbb{R}^{n_p}$ is a \emph{malicious
  disturbance} and
$w_k \in \interval_{w} \triangleq [\underline{w},\overline{w}] \subset
\real^{n_w}$ is bounded process noise.  At time step $k$, every agent
$i \in \nodes$ takes a measurement $y^i_k \in \real^{{m^i}}$, known
only to itself, which is perturbed by
$v^i_k \in \interval_{v}^i \triangleq [\underline{v}^i,\overline{v}^i]
\subset \real^{n_v^i}$, a bounded sensor (measurement) noise signal.
Finally, $A \in \real^{n \times n}$, $B \in \real^{n \times n_w}$,
$G \in \real^{n \times n_p}$, $C^i \in \real^{m^i \times n}$,
$D^i \in \real^{m^i \times n^i_v}$, and
$H^i \in \real^{m^i \times n_p}$ are system matrices known to all
agents, where $\rank(H^i)=r^i$.
Note that no restriction is
made on $H^i$ to be either the zero matrix (no direct feedthrough), or
to have full column rank when there is direct feedthrough.  The
agents' goal is to simultaneously estimate the trajectories of
\eqref{eq:system} as well as the unknown input $d_k$ in a distributed
manner, when states are initialized in an interval
$\interval_x \triangleq [\underline{x}_0,\overline{x}_0] \subset
\real^{n}$, with $\underline{x}_0,\overline{x}_0$ known to all agents.

\noindent \textbf{\emph{Unknown Input Signal Assumptions.}}
We make no assumption about the unknown signal $d_k$, i.e., we require
no prior knowledge such as its distribution, dynamics, or bounds.
\begin{rem}
 {\rm System \eqref{eq:system} can be easily extended to cover the case
  where {different} attack signals $d^s_k \in \mathbb{R}^{n_s}$ and
  $d^o_k \in \mathbb{R}^{n_o}$ with the corresponding matrices
  $\hat{G} \in \mathbb{R}^{n \times n_s}$ and
  $\hat{H}^i \in \mathbb{R}^{{m^i} \times n_o}$ are injected into the
  actuators and sensors, respectively. In this case, courtesy of the
  fact that the unknown input signals can be completely arbitrary, by
  lumping them into a {newly} defined unknown input signal
  $d_k \triangleq [(d^s_k)^\top (d^o_k)^\top]^\top \in
  \mathbb{R}^{n_p}, n_p \triangleq n_o+n_s$, as well as defining
  ${G}\triangleq \begin{bmatrix} \hat{G} & 0_{n \times
      n_o}\end{bmatrix}$,
  ${H}^i\triangleq \begin{bmatrix} 0_{{m^i} \times n_s} &
    \hat{H}^i \end{bmatrix}$, we can equivalently transform the
  considered system to a new representation, precisely in the form of
  \eqref{eq:system}.}
\end{rem}
\scott{%
\begin{defn}[State and Input Framers]\label{defn:framers}
  For an agent $i\in\nodes$, the sequences
  $\{\overline{x}^i_k,\underline{x}^i_k\}_{k \ge 0} \subseteq \real^n$
  are called upper and lower \emph{local state framers} for
  $\mathcal{P}$ if $\underline{x}^i_k \leq x_k \leq \overline{x}^i_k$
  for all $k \ge 0$. Moreover, we define the \emph{local state framer
    errors} as
  \begin{align}\label{eq:error_1}
    \begin{array}{rl}
      \underline{e}^{i}_{k}
      &\triangleq x_k-\underline{x}^i_k, \
        \overline{e}^{i}_{k} \triangleq \overline{x}^i_k-x_k,\ \forall k\geq 0.
    \end{array}
  \end{align}
  Finally, the \emph{collective state framer error} is defined as
  \begin{align}\label{eq:error_2}
    \begin{array}{rl}
      e_{k} &\triangleq  \begin{bmatrix}(\underline{e}^{1}_{k})^\top &
        (\overline{e}^{1}_{k})^\top & \cdots & (\underline{e}^{N}_{k})^\top &
        (\overline{e}^{N}_{k})^\top \end{bmatrix}^\top \in \real^{2Nn}.\\
    \end{array}
  \end{align}
The \emph{local input framers}
$\{\underline{d}^i_k, \overline{d}^i_k\}$, the \emph{local input
  framer errors} $\{\underline{\delta}^i_k, \overline{\delta}^i_k\}$,
and the \emph{collective input framer error} $\delta_k$ are defined
similarly with respect to the unknown input $d_k$.
\end{defn}}%
The problem of designing a distributed resilient state and input
interval observer addressed here is cast as follows:
\scott{
\begin{problem}\label{prob:SISIO}
  Given a multi-agent system, design a distributed resilient interval
  observer for $\mathcal{P}$, i.e., an algorithm that computes uniformly bounded
  state and input framers for $\mathcal{P}$.
\end{problem}}

\section{Proposed Distributed Interval Observer}
\label{sec:observer}
In this section, we describe our novel resilient distributed interval
observer design, its stability, and a tractable distributed procedure
for computing stabilizing observer gains.

\subsection{Distributed Input and State Framer and its Correctness}
Before describing our proposed observer, we first transform the system
into an equivalent representation which decouples the problem of
estimating the state and the adversarial input.
Inspired by the work in \cite{yong2018simultaneous}, we carry
out a singular value decomposition (SVD) on the feedthrough matrix
$H^i$, which decomposes the unknown input signal into two components
$d^i_{1,k}$ and $d^i_{2,k}$. Consequently, we obtain two constituents
for the measurement signal: $z^i_{1,k}$, which is affected by the
unknown input through an invertible feedthrough matrix
$S^i \in \mathbb{R}^{r^i \times r^i}$, and $z^i_{2,k}$, which is not
compromised by the unknown input signal. \kj{Then,} \eqref{eq:system}
can be represented as
\begin{subequations}\label{eq:system2}
  \begin{align}
    x_{k+1} &= Ax_k+Bw_k+G^i_1d^i_{1,k}+G^i_2d^i_{2,k}, \label{eq:system_2_state} \\
    z^i_{1,k} &= C_1^ix_k+D_1^iv^i_{k}+S^id^i_{1,k},\label{eq:system_2_output1} \\
    z^i_{2,k} &= C_2^ix_k+D_2^iv^i_{k}, \label{eq:system_2_output2}\\
    d^i_{k} &= V^{i}_{1}d_{1,k}+V^{i}_{2}d_{2,k}.\label{eq:system_2_input}
  \end{align}
\end{subequations}
To increase readability, the details of the transformation and how to
compute $V^i_1$, $V^i_2$, $G^i_1$, $G^i_2$, $C^i_1$, $C^i_2$, $D^i_1$,
and $D^i_2$ are provided in Appendix A. We
further define $M^i_1 \triangleq (S^i)^{-1}$ \kj{and the concatenated
  noise vector
  $\eta^i_k \triangleq [v^{i\top}_{k-1} \ w^{\top}_{k-1} \
  v^{i\top}_{k}]^\top$ with upper and lower bounds:
\begin{align*}
  \underline{\eta}^i \triangleq [\underline{v}^{i\top} \ \underline{w}^{\top} \ \underline{v}^{i\top}]^\top \ \text{and} \
  \overline{\eta}^i \triangleq [\overline{v}^{i\top} \ \overline{w}^{\top} \ \overline{v}^{i\top}]^\top.
\end{align*}}%
\scott{To obtain input and state estimates, we require that each agent
  has access to adequate measurements which are not compromised by the
  unknown input. The following assumption ensures this.  We refer the
  reader to \cite[Lemma 1]{yong2018simultaneous} for a discussion on
  the necessity of this type of requirement in centralized estimation
  algorithms subject to adversaries. In a nutshell, unless this type
  of observability conditions hold, an adversary can arbitrarily drive
  system components around in a fully stealthy manner.}
\begin{assumption}\label{ass:rank}
  $C^i_2G^i_2$ is full column rank for every $i \in \nodes$.%
  \footnote{\scott{To relax this condition, nodes can receive
      measurements from neighbors so that the assumption holds for the
      concatenated observation matrices.}}
  Hence, there exists $M^i_2=(C^i_2G^i_2)^\dagger$ such that
  $M^i_2C^i_2G^i_2=I$.
\end{assumption}
\scott{It is worth noting that $C^i_2$ and $G_2^i$ are both affine
  transformations of $C$ and $G$, respectively. Moreover, adequacy of
  measurements plays an important role when applying a singular value
  decomposition on the direct feedthrough matrix in the output
  equation (cf. Appendix A for more details). }
\kj{W}e are ready to propose our \kj{recursive} distributed
simultaneous state
and unknown input observer.

\subsection{Distributed Input and State Framer and its Correctness}
\label{sec:obsv}
To address Problem~\ref{prob:SISIO}, we propose a four-step
procedure, summarized as
{\color{black} the \dsisoalgo (\dsiso)} in Algorithm~\ref{alg1}.
\noindent\textbf{\textit{ i) State Propagation and Measurement
    Update:}}
Given $\underline{x}^i_{k}$, $\overline{x}^i_{k}$, $z^i_{1,k}$,
$z^i_{2,k}$, and $z^i_{2,k+1}$, each agent~$i \in \nodes$ performs a
state propagation and a local measurement update step using
to-be-designed observer gains $L^i$ and
$\Gamma^i \in \real^{n \times {m^i}}$:
%
\begin{gather}\label{eq:observer}
  \hspace{-.1cm}
  \begin{split}
    \underline{x}^{i,0}_{k+1}
    & =  \mathtt{A}^{i+}\,\underline{x}_k^i
    -  \mathtt{A}^{i-}\overline{x}_k^i
    + \mathtt{L}^{i+}\underline{\eta}^i
    - \mathtt{L}^{i-}\overline{\eta}^i+  \Psi^i\xi^i_{k+1}, \\
    \overline{x}^{i0}_{k+1}
    & =  \mathtt{A}^{i+}\overline{x}_k^i
    -  \mathtt{A}^{i-}\underline{x}_k^i
    + \mathtt{L}^{i+}\overline{\eta}^i
    - \mathtt{L}^{i-}\underline{\eta}^i+  \Psi^i \xi^i_{k+1},
  \end{split}
\end{gather}
\scott{where $\mathtt{A}^{i}$, $\mathtt{L}^{i}$, and $\Psi^i$, which
define the observer dynamics, depend linearly on $L^i$ and $\Gamma^i$
and are described in Appendix~B, and $\xi^i_{k+1} \triangleq
  \begin{bmatrix}(z^i_{1,k})^\top & (z^i_{2,k})^\top &
(z^i_{2,k+1})^\top \end{bmatrix}^\top.$ We also note that
$\mathtt{A}^{i}$ has the form $\mathtt{A}^{i} \triangleq T^i\hat{A}^i
- L^iC_2^i$, where $T_i = I-\Gamma^iC_2^i$, and $\hat{A}^i$ depends
only on parameters of $\mathcal{P}.$}

\noindent\textbf{\textit{ ii) State Framer Network Update:}}
Each agent~$i$ shares its local state framers with
its neighbors in the network, \mmo{updating them by} taking the
tightest interval from all neighbors via intersection,
\begin{align}\label{eq:state_network_update}
  \underline{x}^{i}_{k} = \max_{j\in\neighbors_i}\underline{x}^{j,0}_k, \quad
  \overline{x}^{i}_{k} = \min_{j\in\neighbors_i}\overline{x}^{j,0}_k.
\end{align}
\scott{
We consider only one iteration of this operation for simplicity; the extension to
multiple iterations is straightforward.}

\noindent\textbf{\textit{ iii) Input Estimation:}}
Given the state framers \eqref{eq:state_network_update},
agent~$i \in \nodes$ leverages the state dynamics and its measurement
of the system to compute local input framers as follows:
\begin{align}\label{eq:input_framers}
  \hspace{-.5em}\begin{array}{rl}
    \underline{d}^{i,0}_k &= \mathtt{A}^{i+}_d \,\underline{x}^i_k
                            - \mathtt{A}_d^{i-}\overline{x}^i_k
                            + \mathtt{F}^{i+}\underline{\eta}^i
                            - \mathtt{F}^{i-}\overline{\eta}^i
                            + \Phi^i\xi^i_{k+1},\\
    \overline{d}^{i,0}_k &= \mathtt{A}_d^{i+}\overline{x}^i_k
                           - \mathtt{A}_d^{i-}\underline{x}^i_k
                           + \mathtt{F}^{i+}\overline{\eta}^i
                           - \mathtt{F}^{i-}\underline{\eta}^i
                           + \Phi^i\xi^i_{k+1}.
  \end{array}
\end{align}
\scott{where $\mathtt{A}_d^i$, $\mathtt{F}^i$, and $\Phi^i$ are described in Appendix B}.

\noindent\textbf{\textit{ iv) Input Framer Network Update:}}
Finally, each agent~$i$ shares its local input
framers with its neighbors in the network, again taking the intersection,
\begin{align}\label{eq:network_update}
    \begin{split}
      \underline{d}^{i}_{k} = \max_{j\in\neighbors_i}\underline{d}^{j,0}_k, \quad
      \overline{d}^{i}_{k} = \min_{j\in\neighbors_i}\overline{d}^{j,0}_k,
    \end{split}
\end{align}
\begin{rem}
  \scott{\rm There are many existing centralized interval observer
    designs in the literature that could potentially be used for step
    i). However, most of these methods rely on similarity
    transformations \cite{mazenc2014interval} which depend on the
    observation matrices $C^i$. In a multi-agent setting, use of these
    methods necessitates transforming to and from the original
    coordinates whenever estimates are shared over the network. Each
    repeated transformation incurs the so-called ``wrapping effect''
    as a result of Proposition~\ref{prop:bounding}, which worsens the
    estimation error and negates the benefit of exchanging information
    over the network. We avoid this with the choice of
    \eqref{eq:observer}, which is computed directly in the original
    coordinates.}
\end{rem}
\begin{algorithm}[]
  \setstretch{1.1}
  \caption{\dsiso \, at node $i$.}
  \label{alg1}
  \begin{algorithmic}[1]
    \renewcommand{\algorithmicrequire}{\textbf{Input:}}
    \renewcommand{\algorithmicensure}{\textbf{Output:}}
    \Require 
    $\underline{x}^i_0$, $\overline{x}^i_0$,; \textbf{Output:}
    $\{\underline{x}^i_{k},\overline{x}^i_{k},\underline{d}^i_{k},\overline{d}^i_{k}\}_{k\ge0}$;
    \State Compute $L^i$ $\Gamma^i$,
    and $T^i$ by solving \eqref{eq:stab_LMI};
    \State $k \gets 1$
    \Loop
      \Statex \scott{$\triangleright$ \textbf{State propagation and measurement update}}
      \State Compute $\underline{x}^{i,0}_{k}$ and
      $\overline{x}^{i,0}_{k}$ using \eqref{eq:observer};
      \Statex \scott{$\triangleright$ \textbf{State framer network update}}
      \State Send $\underline{x}^{i,0}_{k}$ and
      $\overline{x}^{i,0}_{k}$ to $\{j \ : \ i \in \neighbors_j\}$;
      \State Receive $\underline{x}^{j,0}_{k}$ and
      $\overline{x}^{j,0}_{k}$ from $j \in \neighbors_i$;
      \State $\displaystyle\underline{x}^{i}_{k} \gets\,
      \max_{j\in\neighbors_i}\underline{x}^{j,0}_{k}; \quad
      \overline{x}^{i}_{k} \gets\, \min_{j\in\neighbors_i}\overline{x}^{j,0}_{k};$
      \Statex \scott{$\triangleright$ \textbf{Input framer estimation}}
      \State Compute $\underline{d}^{i,0}_{k}$ and
      $\overline{d}^{i,0}_{k}$ using \eqref{eq:input_framers};
      \Statex \scott{$\triangleright$ \textbf{Input framer network update}}
      \State Send $\underline{d}^{i,0}_{k}$ and
      $\overline{d}^{i,0}_{k}$ to $\{j \ : \ i \in \neighbors_j\}$;
      \State Receive $\underline{d}^{j,0}_{k}$ and
      $\overline{d}^{j,0}_{k}$ from $j \in \neighbors_i$;
      \State     $\displaystyle\underline{d}^{i}_{k} \gets\,
      \max_{j\in\neighbors_i}\underline{d}^{j,0}_{k}; \quad
      \overline{d}^{i}_{k} \gets\, \min_{j\in\neighbors_i}\overline{d}^{j,0}_{k};$
      \State $k \gets k+1$;
    \EndLoop
    \Return $\{\underline{x}^i_{k},\overline{x}^i_{k},\underline{d}^i_{k},\overline{d}^i_{k}\}_{k\ge0}$
  \end{algorithmic}
\end{algorithm}
\begin{lem}[Distributed Framer Construction]
  \label{lem:ind_framer}
  The \dsiso\ algorithm outputs interval state and input framers for
  $\mathcal{P}$.
\end{lem}
\begin{proof}
  \scott{See Appendix B.}
\end{proof}
\subsection{Distributed Stabilizing and Error Minimization}\label{subsec:CISS}
In this subsection, we investigate conditions on the observer gains
$L^i$ and $\Gamma^i$, as well as the communication graph
$\graph$, that lead to an \emph{input-to-state stable} (ISS)\footnote{The reader is referred to \cite[Definition 2]{khajenejad2022distributed} for a detailed description of an ISS interval observer.} distributed observer, which equivalently
results in a uniformly bounded observer error sequence given in
\eqref{eq:error_1}--\eqref{eq:error_2}, in the presence of bounded
noise. To guarantee stability, we use the following assumption on the agents' observation matrices and the structure of the network graph.
\begin{assumption}[Collective Positive Detectability over
  Neighborhoods (CPDN) \cite{khajenejad2022distributed}]\label{ass:col_pos_det_neigh}
  For every state dimension $s \in \until{n}$ and every agent
  $i\in\nodes$, there is an agent $\ell(i,s)\in\neighbors_i$ such that
  there exist gains $L^{\ell(i,s)}$, and
  $\Gamma^{\ell(i,s)}$ satisfying
  \begin{align*}
  \|(T^{\ell(i,s)}\hat{A}^{\ell(i,s)}-L^{\ell(i,s)}C_2^{\ell(i,s)})_s\|_1 < 1.
  \end{align*}
\end{assumption}
\scott{Intuitively, this assumption narrows the problem of stability
  to subgraphs. Within these subgraphs, we require that for each state
  dimension $s$, there is a node that, given estimates of all other
  state dimensions $\{1,\dots,s-1,s+1,\dots,n\}$, can compute an
  accurate estimate of dimension~$s$. With this assumption in mind, we
  propose a two-step process to design the observer gains $L^i$ and
  $\Gamma^i$. First, each node executes a procedure
  (Algorithm~\ref{alg2}) which \emph{verifies}
  Assumption~\ref{ass:col_pos_det_neigh}, returning false if
  it is not satisfied, or else \emph{computes} some feasible
  stabilizing gains and the set of state dimensions which a node can
  contribute to estimating, i.e.,
\begin{align*}
  \mathbb{J}^i \triangleq \{s \ : \ \exists T^i, L^i, \Gamma^i \text{ s.t. } \|(T^i\mathtt{A}^{i} - L^iC_2^i)_s \|_1 < 1 \}.
\end{align*}
Then, given this information, each node solves the LP in
\eqref{eq:stab_LMI}, which simultaneously \emph{guarantees stability}
and minimizes an upper bound on the observer error.} This design
process improves on the one proposed in
\cite{khajenejad2022distributed} in the absence of attacks, by first
identifying ``stabilizing" agents for each state dimension, then
minimizing an upper bound on the error while enforcing the stability
condition. In this way, the design includes a sense of noise/error
attenuation. \kj{The following theorem formalizes our main results on
how to tractably synthesize stabilizing and error minimizing observer
gains in a distributed manner. }
\begin{thm}[Distributed Input and State Interval Observer Design]
  \label{thm:suff_stability}
  {Suppose} Assumptions \ref{ass:rank} and \ref{ass:col_pos_det_neigh} hold and $L^i_*$, $T^i_*$, and $\Gamma^i_*$ are solutions to the \kj{following convex} program:
    \begin{align}\label{eq:stab_LMI}
      \begin{array}{rl}
        \min\limits_{Z^i,L^i,T^i,\Gamma^i} &\quad
         \||\mathtt{L}^{i}|(\overline{\eta}^i-\underline{\eta}^i)\|_\infty \\
        \mathrm{s.t.}\quad &\quad \hspace{-.3cm}T^i = I - \Gamma^iC_2^i,\\
        &\quad \hspace{-.3cm} \sum_{t=1}^n Z^i_{jt} < 1, \ \forall j \in \mathbb{J}^i,\\
       &\quad -Z^i \leq T^i\hat A^i-L^iC_2^i \leq Z^i,\\
      \end{array}
    \end{align}
      where $\mathtt{L}^{i}$ is defined in \eqref{eq:observer} and $\mathbb{J}^i$ is calculated using Algorithm \ref{alg2}. 
    Then, the \kj{\dsiso} algorithm, i.e., \eqref{eq:observer}--\eqref{eq:network_update}, with the corresponding observer gains $L^i_*, T^i_*, \Gamma^i_*$
    constructs an ISS distributed input and state interval observer. Moreover, the steady state observer errors are guaranteed to be bounded:
  \begin{align}\label{eq:ss_upperbounds}
    \hspace{-.5cm}
    \begin{array}{rl}
    \|e^x_{k}\|_{\infty} &\leq \frac{1}{1-\rho_*}
                             \max_{i}
                             \||\mathtt{L}^{i}|
                             \Delta_{\eta}^i\|_\infty, \\
                               \|\delta_k\|_\infty &\leq \frac{\rho(\mathcal{A}_d)}{1-\rho_*} \max_{i}
                             \||\mathtt{L}^{i}|
                             \Delta_{\eta}^i\|_\infty + \max_i \||\mathtt{F}^{i}|
                         \Delta_{\eta}^i\|_\infty,
                         \end{array}
    \end{align}
where $\rho_*$, $\mathcal{A}_d$, and $\Delta_{\eta}^i$ are given in Lemma \ref{lem:errors}.
\end{thm}
\kj{In order to prove our main results in Theorem
  \ref{thm:suff_stability}, we need to first take two intermediate
  steps on i) providing closed form expressions for the observer
  errors and their upper bounds, and ii) ensuring the existence of
  stabilizing gains, stated via Lemmas \ref{lem:errors} and
  \ref{lem:stability}, respectively.}  \scott{To begin, we note that
equations \eqref{eq:observer}-\eqref{eq:network_update} result in a
\emph{switched} linear system, with the following set of possible
switching signals:
\begin{align*}
  \mathcal{M} \triangleq \left\{M \in \{0,1\}^{2Nn\times2Nn} \ :
  \begin{array}{c} M_{ij} = 0, \forall j \notin \neighbors_i \\
    \sum_{k=1}^{2Nn} M_{ik} = 1 \end{array} \right\},
\end{align*}
which encodes all possible permutations of the operation \eqref{eq:state_network_update}.}
\begin{lem}[Error Bounds]\label{lem:errors}
  \scott{For all $\mathcal{B}\in \mathcal{M}$}, the errors of the {\dsiso} observer
  \eqref{eq:observer}--\eqref{eq:network_update} are upper
  bounded as follows:
    \begin{align}\label{eq:error_upper_bounds}
    \begin{array}{rl}
      \|e^x_{k}\|_{\infty} &\leq \|e_{0}\|_{\infty}\rho_*^k
                             + \frac{1-\rho_*^k}{1-\rho_*}
                             \max_{i}
                             \||\mathtt{L}^{i}|
                             \Delta_{\eta}^i\|_\infty, \\
      \|\delta_k\|_\infty &\leq \rho(\mathcal{A}_d)\|e_k\|_\infty + \max_i \||\mathtt{F}^{i}|
                         \Delta_{\eta}^i\|_\infty,
                         \end{array}
    \end{align}
  where
$\rho_* \triangleq \rho(\mathcal{B}\mathcal{A}_x)$, \kj{$\mathcal{A}_s \triangleq \diag{\mathbb{A}^1_s, \dots \mathbb{A}^N_s}$}, $  \mathbb{A}^i_s \triangleq
  \begin{bmatrix}
    \mathtt{A}_s^{i+} & \mathtt{A}_s^{i-} \\
    \mathtt{A}_s^{i-} & \mathtt{A}_s^{i+}
  \end{bmatrix},$ $\mathtt{A}^i_x \triangleq \mathtt{A}^i$, $s \in \{x,d\}$, and
  $\Delta_{\eta}^i \triangleq
\overline{\eta}^i-\underline{\eta}^i.$
%
\end{lem}
\begin{proof}
  \kj{Starting from \eqref{eq:observer}-\eqref{eq:network_update}}
  and following the lines of the proof of \cite[Theorem
  1]{khajenejad2022distributed}, \scott{for any $\mathcal{B} \in \mathcal{M}$}, the framer error\kj{s} can be bounded
  by the positive linear comparison system
  \begin{align}\label{eq:comparison}
    e_{k+1} \leq \mathcal{B}\mathcal{A}_xe_k +
    \mathcal{B}\gamma^x_k, \quad
    \delta_{k} \leq  \mathcal{A}_d e_{k} + \gamma^d_{k},
  \end{align}
  where $\forall s \in \{x,d\}: \gamma^s_k \triangleq [(\lambda^{1}_{s,k})^\top\dots \ (\lambda^{N}_{s,k})^\top]^\top$, $\lambda_{s,k}^i \triangleq
  \begin{bmatrix}
    ((\Xi^i_s)^+\underline{e}^i_{\eta,k}
    + (\Xi^i_s)^-\overline{e}^i_{\eta,k})^\top &
    ((\Xi^i_s)^-\underline{e}^i_{\eta,k}
    + (\Xi^i_s)^+\overline{e}^i_{\eta,k})^\top
  \end{bmatrix}^\top$, $\Xi^i_x \triangleq \mathtt{L}^{i},\Xi^i_d
  \triangleq \mathtt{F}^{i}$, $\underline{e}^i_{\eta,k} \triangleq
  \eta^i_k-\underline{\eta}^i_k$ and $\overline{e}^i_{\eta,k}
  \triangleq
  \overline{\eta}^i_k-\eta^i_k.$ \kj{Moreover,} It follows from the
  solution of
  \eqref{eq:comparison} \kj{that}
  \begin{align}\label{eq:comparison_bounding}
    e_k \leq (\mathcal{B}\mathcal{A}_x)^{k-1} e_0
    + \sum_{s=1}^{k-1}(\mathcal{B}\mathcal{A}_x)^{k-s}\gamma_{s-1}.
  \end{align}
  Further,
  $\|\gamma^s_k\|_{\infty} \leq
  \max_{i}|\Xi^i_s|(\overline{\eta}^i -
  \underline{\eta}^i)$ by non-negativity of $(\Xi^i_s)^+$,
  $(\Xi^i_s)^-$, and
  $\overline{\eta}^i - \underline{\eta}^i$. The result follows
  from \eqref{eq:comparison}, \eqref{eq:comparison_bounding}, sub-multiplicativity of norms and the triangle
  inequality.
\end{proof}
\begin{lem}
  \label{lem:stability}
  If Assumption 2 holds, then \scott{there exist $L^i$ and $\Gamma^i$ such that, for some $\mathcal{B}_* \in \mathcal{M}$, \kj{$\mathcal{B}_*\mathcal{A}_x$ is Schur stable, i.e., $\rho(\mathcal{B}_*\mathcal{A}_x) < 1$}. Consequently,} the {\dsiso} algorithm is ISS.
\end{lem}
\begin{proof}
  \kj{It follows from combining \cite[Theorems 1 $\&$ 2]{khajenejad2022distributed}}.
\end{proof}
We are ready to provide a proof for Theorem \ref{thm:suff_stability} as follows.\\
\textit{Proof of Theorem \ref{thm:suff_stability}.}
\scott{By Lemma~\ref{lem:stability}, Assumption 2 implies the
  existence of gains that render the {\dsiso} algorithm ISS.  It
  remains to show that the solutions of \eqref{eq:stab_LMI} are
  stabilizing.  First, notice that Algorithm \ref{alg2} computes
  $\mathbb{J}^i$ \kj{by solving \eqref{eq:agent_identification}.}
  The use of $\mathbb{J}^i$ in the constraints of
  \eqref{eq:stab_LMI} guarantees that the optimization problem is
  feasible. Furthermore, we can show that since
  Assumption~\ref{ass:col_pos_det_neigh} holds, there exists
  $\mathcal{B}_*$ such that
  $\rho(\mathcal{B}_*\mathcal{A}_x) < 1$, and therefore that the
  {\dsiso} algorithm is ISS. We refer the reader to \cite[Theorem
  2]{khajenejad2022distributed} for the construction of
  $\mathcal{B}_*$.} This in combination with Lemmas \ref{lem:errors}
and \ref{lem:stability} ensures that the bounds in
\eqref{eq:error_upper_bounds} converge to their steady state values
in \eqref{eq:ss_upperbounds}. \qed
\begin{rem}{\rm
    It is worth noting that by the following change of variables, the
    convex program in \eqref{thm:suff_stability} can be easily and
    equivalently stated in the form of a linear program (LP):
    \begin{align*}
      \begin{array}{rll}
        \min\limits_{Z^i,L^i,T^i,\Gamma^i,\eta,\theta} &\quad
                                                         \lambda \\
        \mathrm{s.t.}&\quad
                       \theta(\overline{\eta}^i-\underline{\eta}^i) \leq \lambda \mathbf{1}, \quad
                       -\theta \leq \mathtt{L}^{i} \leq \theta,\\
                                                       &\quad T^i = I - \Gamma^iC_2^i, \quad \sum_{t=1}^n Z^i_{jt} < 1, \ \forall j \in \mathbb{J}^i,\\
                                                       &\quad -Z^i \leq T^i\hat A^i-L^iC_2^i \leq Z^i,\\
      \end{array}
    \end{align*}
    which can be considered as complementary to
    $\mathcal{H}_{\infty}$-optimal observer design, e.g., in
    \cite{khajenejad2019simultaneous,khajenejadasimultaneous,khajenejad2022resilient}.}
\end{rem}
\begin{algorithm}
  \setstretch{1.1}
  \caption{{\dsiso} initialization at node $i$.}
  \label{alg2}
  \begin{algorithmic}[1]
    \renewcommand{\algorithmicrequire}{\textbf{Input:}}
    \renewcommand{\algorithmicensure}{\textbf{Output:}}
    \Require $A$, $C^i$, $\neighbors_i$;
    \textbf{Output:} $\mathbb{J}^i$
    \State Compute $L_*^{i}$, $\Gamma_*^{i}$, and $Z_*^{i}$ by solving
    \begin{align}\label{eq:agent_identification}
      \begin{array}{rl}
        \hspace{-.4cm}\textstyle\min_{Z^i,L^i,\Gamma^i} &\quad  \textstyle\sum_{s=1}^n\sum_{t=1}^n (Z^i)_{st} \\
        \mathrm{s.t.}&\  -Z^i \leq (I - \Gamma^iC_2^i)\hat A^i-L^iC_2^i \leq Z^i.
      \end{array}
    \end{align}
    \State $\mathbb{J}^i \gets \{s \ : \ \sum_{t=1}^n(Z_*^i)_{st} < 1\}$;
    \State $\mathcal{Q}_i \gets \{(I - \Gamma^i_*C_2^i)\hat A^i - L_*^{i}C_2^{i}\}$;
    \State Receive $\mathcal{Q}_j$ from $j\in\neighbors_i$;
    \State $\mathcal{Q}_i \gets \bigcup_{j\in\neighbors_i} Q_j$;
    \If{$\forall s \in \until{n}$, $\exists P \in \mathcal{Q}_i$ s.t. $\|(P)_s\|_1 < 1$}
    \Return false (i.e., Assumption~\ref{ass:col_pos_det_neigh} not satisfied)
    \Else
    \Return $\mathbb{J}^i$
    \EndIf
  \end{algorithmic}
\end{algorithm}

\section{Simulation}\label{sec:example}
In this section we demonstrate the {\dsiso} algorithm on an IEEE 14-bus
system \cite{ieee14bus}. We refer the reader to \cite{Pasqualetti.2013} for the derivation of
the LTI representation of the system, which can be discretized and
written in the form of \eqref{eq:system}.  The $n=10$ dimensional
state
$x_k^\top = \begin{bmatrix}\delta_k^\top &
  \omega_k^\top\end{bmatrix}^\top$ represents the rotor angle and
frequency of each of the 5 generators. Each bus in the test case
corresponds to a node in the algorithm. \scott{The noise signals satisfy $\|w_k\|_\infty < 5$ and $\|v^i_k\| < 1\times 10^{-4}$ $\forall i \in \nodes$}.  Similarly to the example in
\cite{Pasqualetti.2013}, each node (bus) measures its own real power
injection/consumption, the real power flow across all branches
connected to the node, and for generating nodes, the rotor angle of
the associated generator.

In this example, we assume that the generator at node 1 is insecure
and potentially subject to attacks. \scott{Due to the reduction
necessary to eliminate the algebraic constraints of the power system
model \cite{Pasqualetti.2013}, the disturbance appears directly in the
measurements of all nodes, resulting in nonzero $H^i$
matrices. We compute the gains $L^i$ and $\Gamma^i$ by solving
  \eqref{eq:stab_LMI}}. Figures \ref{fig:x} and \ref{fig:d}
show the input and
state framers, respectively. It is clear that the algorithm is able to
estimate the state $x_1$ despite the disturbance with only minor
performance degradation. The switching due to
\eqref{eq:state_network_update}, which depends on the noise, is also
evident. The estimation performance for the other states is
comparatively better, since they are only affected by (known) bounded
noise. Furthermore, all agents are able to maintain an accurate
estimate of the disturbance.

\begin{figure}
  \centering
  \includegraphics[width=0.98\columnwidth]{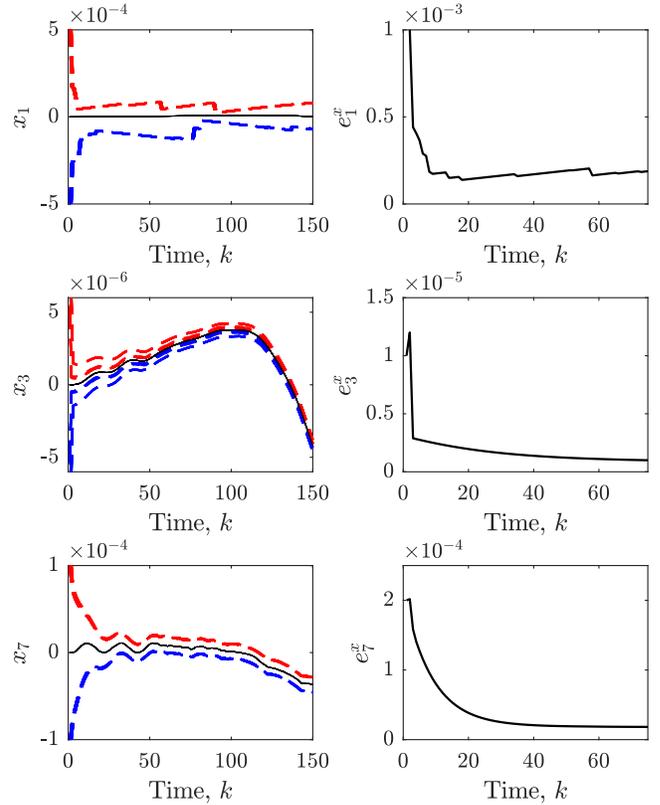}
  \caption{\small State \kj{upper (red) and lower (blue)} framers from all agents \kj{and the true state value (black)} for
    selected states $x_1$, $x_3$, and $x_7$ \kj{(left), as well as  state error interval widths (right).}}
  \label{fig:x}
\end{figure}
\begin{figure}[]
  \centering
  \includegraphics[width=0.98\columnwidth]{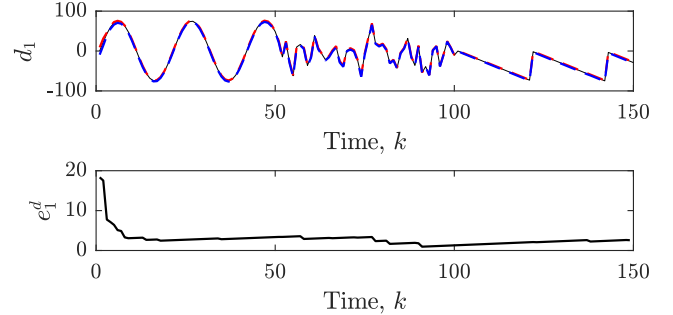}
  \caption{\small Unknown input framers and framer errors from all agents. Blue and red denote lower and upper framers, respectively, and black is the true signal.}
  \label{fig:d}
\end{figure}
\section{Conclusion and Future Work} \label{sec:conclusion} This paper
introduced a novel distributed algorithm for interval estimates of
states and unknown inputs for LTI systems with bounded noise, whose
sensors and actuators are compromised by false data injection
attacks. Without imposing any restrictive assumptions such as
boundedness or stochasticity on the unknown input (attack) signals, we
addressed the correctness of the proposed distributed observer, and
moreover, \scott{analyzed the stability of the observer by considering the
switched linear dynamics of the resulting error system}. Finally,
we provided a tractable method for computing stabilizing gains which
aim to minimize the steady state \kj{input and state} error of the
observer. \kj{Hence, our approach can serve the purpose of resilient
  estimation in bounded-error networked cyber-physical systems. In the
  future, we consider extending our approach to nonlinear and hybrid
  systems, as well as including other type of adversarial effects such
  as switching and network attacks in our setting.}  \setstretch{0.96}
\section*{Appendix}
\subsection{Similarity Transformation}\label{sec:transformation}
Let $r^i\triangleq \rank{(H^i)}$. Using SVD, 
$H^i= \begin{bmatrix}U^i_{1}& U^i_{2} \end{bmatrix} \begin{bmatrix} S^i & 0 \\ 0 & 0 \end{bmatrix} \begin{bmatrix} (V_{1}^i)^ \top \\ (V_{2}^i)^ \top \end{bmatrix}$,
where $S^i \in \mathbb{R}^{r^i \times r^i}$ is a diagonal matrix of full rank, $U_{1} \in \mathbb{R}^{m^i \times r^i}$, $U_{2} \in \mathbb{R}^{m^i \times (m^i-r^i)}$, $V_{1} \in \mathbb{R}^{n_p \times r^i}$ and $V_{2} \in \mathbb{R}^{n_p \times (n_p-r^i)}$, while $U\triangleq \begin{bmatrix} U_{1} & U_{2} \end{bmatrix}$ and $V\triangleq \begin{bmatrix} V_{1} & V_{2} \end{bmatrix}$ are unitary matrices. Then, $d_k$ 
can be decoupled into two orthogonal components as: 
$d^i_{1,k}=(V_{1}^i)^\top d_k,
d^i_{2,k}=(V_{2}^i)^\top d_k,d_k =V^i_{1} d^i_{1,k}+V^i_{2} d^i_{2,k},$ which transforms the system dynamics \eqref{eq:system} to the representation in \eqref{eq:system2}, where $G^i_{1} \triangleq G V^i_{1}$, $G^i_{2} \triangleq G V^i_{2}$, $H^i_{1} \triangleq H^i V^i_{1}=U^i_{1} S^i,{C}^i_{1} \triangleq (U^i_1)^\top {C}^i, {C}^i_{2} \triangleq (U^i_{2})^\top {C}^i, {D}^i_{1} \triangleq (U^i_{1})^\top {D}^i$, and ${D}^i_{2} \triangleq  (U^i_{2})^\top {D}^i$.

\subsection{Proof of Lemma \ref{lem:ind_framer}}\label{sec:proof}
First, note that \eqref{eq:system_2_output1} implies that
\begin{align}\label{eq:system_2_output11}
  d^i_{1,k}=M^i_1(z^i_{1,k} - C_1^ix_k - D_1^iv^i_{k}).
\end{align}
This, in combination with \eqref{eq:system_2_output2} and
\eqref{eq:system_2_state} results in
\begin{align*}
  M^i_2z^i_{2,k+1} &= M^i_2(C_2^ix_{k+1}+D_2^iv^i_{k+1}) \\
                   &=  M^i_2(C_2^i(Ax_k+Bw_k
                     + G^i_1(M^i_1(z^i_{1,k} \\
                   &\quad - C_1^ix_k - D_1^iv^i_{k})
                     + G^i_2d^i_{2,k})
                     + D_2^iv^i_{k+1}),
\end{align*}
which given Assumption \ref{ass:rank}, returns
\tightmath{
  \begin{align}\label{eq:system_2_output22}
    \hspace{-.3cm}d^i_{2,k} = M^i_2(z^i_{k+1} - C^i_2(G^i_1M^i_1z^i_{1,k} + Q^ix_k)) + E^i\eta_{k+1},
  \end{align}}%
where $E^i = M_2^i\begin{bmatrix}C_2^iG_1^iM_1^iD_1^i & -C_2^iB & -D_2^i \end{bmatrix}$. By plugging $d_{1,k}^i$ and $d_{2,k}^i$ from \eqref{eq:system_2_output11} and \eqref{eq:system_2_output22}
into \eqref{eq:system_2_state},
\tightmath{
  \begin{align*}
    \hspace{-.3cm}
    x_{k+1} &= J^iQ^i x_k +
              \tilde{E}^i \eta^i_{k+1}
              + J^iG^i_1M^i_1z^i_{1,k}
              + G^i_2M^i_2C^i_2z^i_{2,k+1},
  \end{align*}}%
where $\tilde{E} = \begin{bmatrix}-J^iG_1^iM_1^iD_1^i & J^iB & -G_2^iM_2^iD_2^i\end{bmatrix}$.
Combined with the fact that
$T^i =I-\Gamma^iC^i_2$, this implies
\tightmath{
  \begin{align}\label{eq:interm_1}
    x_{k+1} 
            =T^i(J^iQ^ix_k+ \tilde{E}^i \eta^i_{k+1}
              + \hat z^i_{k+1})
              + \Gamma^iC_2^ix_{k+1}.
  \end{align}}%
  Plugging in $C^i_2x_{k+1}=z^i_{2,k+1}-D^i_2v^i_{k+1}$
  from \eqref{eq:system_2_output2} and adding
  the \emph{zero term} $L^i(z^i_{2,k}-C^i_2x_k-D^i_2v^i_k)$ to the right
  hand side of \eqref{eq:interm_1}, then collecting like terms,  results in
  \begin{align}\label{eq:sys_equiv}
    x_{k+1}=\mathtt{A}^{i}x_k+\mathtt{L}^{i} \eta^i_{k+1}+ \Psi^i\xi^i_{k+1}, 
  \end{align}
  By applying Proposition \ref{prop:bounding} to
  all the uncertain terms in the right hand side of
  \eqref{eq:sys_equiv}:
    $\underline{x}^{i}_{k} \hspace{-.15cm}\leq\hspace{-.15cm} x_{k} \hspace{-.15cm}\leq\hspace{-.15cm} \overline{x}^{i}_{k}
    \hspace{-.15cm}\Rightarrow \hspace{-.15cm}\underline{x}^{i,0}_{k+1} \hspace{-.15cm}\leq\hspace{-.15cm} x_{k+1}\hspace{-.15cm} \leq\hspace{-.15cm}
    \overline{x}^{i,0}_{k+1}$,
  where $\underline{x}^{i,0}_{k+1},\overline{x}^{i,0}_{k+1}$ are given
  in \eqref{eq:observer}. This means that individual framers/interval
  estimates are correct. When the framer condition is satisfied for
  all nodes, the intersection of all the individual estimates of
  neighboring nodes (cf. \eqref{eq:state_network_update}) also results
  in correct interval framers, i.e.
  \begin{align*}
    \underline{x}^{i,0}_{k} \leq x_{k} \leq \overline{x}^{i,0}_{k},
    \; \forall i \in \nodes
    \implies \underline{x}^{i}_{k} \leq x_{k} \leq \overline{x}^{i}_{k},
    \; \forall i \in \nodes.
  \end{align*}
  Furthermore, plugging $d^i_{1,k}$ and $d^i_{2,k}$ from
  \eqref{eq:system_2_output11} and \eqref{eq:system_2_output22} into
  \eqref{eq:system_2_input} and applying Proposition
  \ref{prop:bounding} returns the input framers in
  \eqref{eq:input_framers}, where their intersection is still a framer
  (cf. \eqref{eq:network_update}) by the same reasoning as for the
  state framers.  Since the initial state framers are known to
  all~$i$, by induction~\eqref{eq:observer}-\eqref{eq:network_update}
  constructs a correct distributed interval state and input framer
  for~\eqref{eq:system}.

\bibliographystyle{unsrturl}
{\tiny
  \bibliography{biblio}
}

\end{document}